\documentclass[english,a4paper,11pt]{article}

\usepackage{amsmath, amsxtra, amsfonts, amssymb, amstext, amsthm}
\usepackage{booktabs}
\usepackage{graphics,color}
\usepackage{babel}
\usepackage[noadjust]{cite}
\usepackage{url}
\usepackage[algo2e,ruled,vlined,linesnumbered]{algorithm2e}
\usepackage{todonotes}
\usepackage{fullpage}

\renewcommand{\epsilon}{{\varepsilon}}
\newcommand{\R}{{\mathbb{R}}}
\newcommand{\N}{{\mathbb{N}}}

\newcommand{\e}{\textup{e}}
\newcommand{\eps}{{\varepsilon}}
\newcommand{\ignore}[1]{}

\newcommand{\cR}{{\mathcal{R}}}

\DeclareMathOperator{\eq}{eq}
\DeclareMathOperator{\col}{col}

\newtheorem{theorem}{Theorem}
\newtheorem{lemma}[theorem]{Lemma}

\newtheorem{corollary}[theorem]{Corollary}
\newtheorem{claim}[theorem]{Claim}

\newcommand{\assign}{\leftarrow}

\newboolean{hasLongCites}
\ifthenelse{\boolean{hasLongCites}} {
\newcommand{\ifLongCites}[2]{#1}
} { 
\newcommand{\ifLongCites}[2]{#2}
}

\begin{document}

\title{Playing Mastermind with Many Colors}

\author{Benjamin Doerr$^1$ \and Carola Doerr$^{1,2}$ \and Reto Sp\"{o}hel$^1$ \and Henning Thomas$^3$}

\date{$^1$Max-Planck-Institut f{\"u}r Informatik, 66123 Saarbr{\"u}cken, Germany \\
$^2$Universit\'{e} Paris Diderot, LIAFA, Paris, France\\
$^3$Institute of Theoretical Computer Science, ETH Z\"urich, 8092 Z\"urich, Switzerland}

\maketitle

\begin{abstract}
We analyze the general version of the classic guessing game  Mastermind with $n$~positions and $k$~colors. Since the case $k \le n^{1-\varepsilon}$, $\varepsilon>0$ a constant, is well understood, we concentrate on larger numbers of colors. For the most prominent case $k = n$, our results imply that Codebreaker can find the secret code with $O(n \log \log n)$ guesses. This bound is valid also when only black answer-pegs are used. It improves the $O(n \log n)$ bound first proven by Chv\'atal (Combinatorica 3 (1983), 325--329). We also show that if both black and white answer-pegs are used, then the $O(n \log\log n)$ bound holds for up to $n^2 \log\log n$ colors. These bounds are almost tight as the known lower bound of $\Omega(n)$ shows. Unlike for $k \le n^{1-\varepsilon}$, simply guessing at random until the secret code is determined is not sufficient. In fact, we show that an optimal non-adaptive strategy (deterministic or randomized) needs $\Theta(n \log n)$ guesses.
\end{abstract}

\textbf{Category:} F.2.2 [Analysis of Algorithms and Problem Complexity]: Nonnumerical Algorithms and Problems


\maketitle

\sloppy{
\section{Introduction}

Mastermind (see Section~\ref{sec:rules} for the rules) and other guessing games like liar games~\cite{liarsurvey,Sp94} have attracted the attention of computer scientists not only because of their playful nature, but more importantly because of their relation to fundamental complexity and information-theoretic questions. In fact, Mastermind with two colors was first analyzed \ifLongCites{in \cite{Erd63}}{by Erd\H os and R\'enyi~\cite{Erd63} in 1963}, several years before the release of Mastermind as a commercial boardgame.

Since then, intensive research by various scientific communities produced a plethora of results on various aspects of the Mastermind game (see also the literature review in Section~\ref{sec:review}). In a famous 1983 paper, 
\ifLongCites{\cite{Chvatal83}}{Chv\'atal~\cite{Chvatal83}} determined, precisely up to constant factors, the asymptotic number of queries needed on a board of size $n$ for all numbers $k$ of colors with $k \le n^{1-\eps}$, $\eps>0$ a constant.
Interestingly, a very simple guessing strategy suffices, namely asking random guesses until the answers uniquely determine the secret code. 

Surprisingly, for larger numbers of colors, no sharp bounds exist. In particular for the natural case of $n$ positions and $k = n$ colors, Chv\'atal's bounds $O(n \log n)$ and $\Omega(n)$ from 1983 are still the best known asymptotic results. 


We almost close this gap open for roughly 30 years and prove that Codebreaker can solve the $k = n$ game using only $O(n \log\log n)$ guesses. 
This bound, as Chv\'atal's, even holds for black-pegs only Mastermind.
When also white answer-pegs are used, we obtain a similar improvement from the previous-best $O(n \log n)$ bound to $O(n \log\log n)$ for all $n \le k \le n^2 \log\log n$. 


\subsection{Mastermind}\label{sec:rules}

Mastermind is a two-player board game invented in the seventies by the Israeli telecommunication expert Mordechai Meirowitz. 
The first player, called Codemaker here, privately chooses a color combination of four pegs. Each peg can be chosen from a set of six colors. 
The goal of the second player, Codebreaker, is to identify this \emph{secret code.} 
To do so, he guesses arbitrary length-4 color combinations. 
For each such guess he receives information of how close his guess is to Codemaker's secret code. Codebreaker's aim is to use as few guesses as possible.

\begin{figure}[t]
	\begin{center}
\includegraphics[scale=.37]{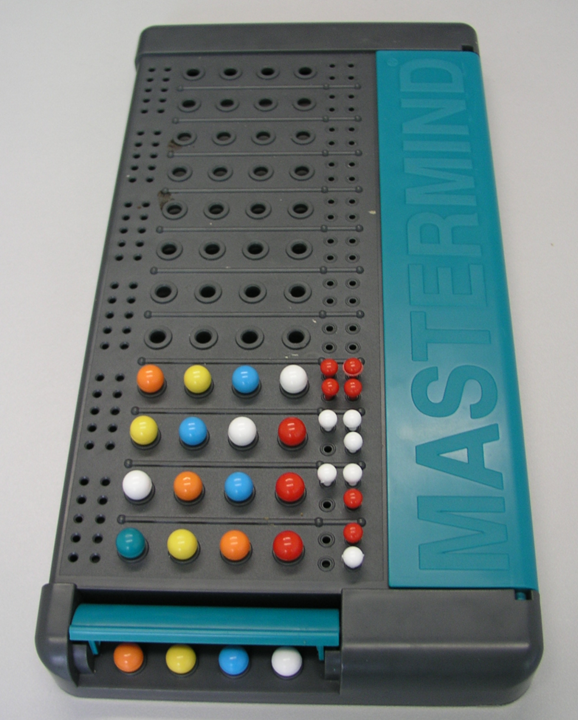}
	\end{center}
\caption{A typical round of Mastermind}
\label{fig}
\end{figure}

Besides the original 4-position 6-color Mastermind game, various versions  with other numbers of positions or colors are commercially available. The scientific community, naturally, often regards a generalized version with $n$ positions and $k$ colors (according to \ifLongCites{\cite{Chvatal83}}{Chv\'atal~\cite{Chvatal83}}, this was first suggested by Pierre Duchet). For a precise description of this game, let us denote by $[k]$ the set $\{1, \ldots, k \}$ of positive integers not exceeding $k$. At the start of the game, Codemaker chooses a secret code $z \in [k]^n$. In each round, Codebreaker guesses a string $x \in [k]^n$. Codemaker replies with the numbers 
$\eq(z,x):= |\{ i \in [n] \mid z_i=x_i\}|$
of positions in which her and Codebreaker's string coincide, and with 
$\pi(z,x)$, the number of additional pegs having the right color, but being in the wrong position. Formally, $\pi(z,x):=\max_{\rho \in S_n}|\{i \in [n] \mid z_{i}=x_{\rho(i)}\}|-\eq(z,x)$, where $S_n$ denotes the set of all permutations of the set $[n]$. In the original game, $\eq(z,x)$ is indicated by black answer-pegs, and $\pi(z,x)$ is indicated by white answer-pegs.
Based on this and all previous answers, Codebreaker may choose his next guess. He ``wins'' the game if his guess equals Codemaker's secret code. 

We should note that often, and partially also in this work, a black-pegs only variant is studied, in which Codemaker reveals $\eq(z,x)$ but not $\pi(z,x)$. This is justified both by several applications (see Section~\ref{sec:review}) and by the insight that, in particular for small numbers of colors, the white answer-pegs do not significantly improve Codebreaker's situation (see Section~\ref{sec:black-and-white}).

\subsection{Previous Results}\label{sec:previous}

Mastermind has been studied intensively in the mathematics and computer science literature. For the original 4-position 6-color version, \ifLongCites{\cite{Knuth77}}{Knuth~\cite{Knuth77}} has given a deterministic strategy that wins the game in at most five guesses. He also showed that no deterministic strategy has a 4-round guarantee.

The generalized $n$-position $k$-color version was investigated by \ifLongCites{\cite{Chvatal83}}{Chv\'atal~\cite{Chvatal83}}. He 
noted that a simple information-theoretic argument (attributed to Pierre Duchet) provides a lower bound of $\Omega(n \log k / \log n)$ for any $k=k(n)$.  

Extending the result\ifLongCites{~\cite{Erd63}}{~\cite{Erd63} of Erd\H os and R\'enyi} from $k=2$ to larger numbers of colors, he then showed that for any fixed $\eps > 0$, $n$ sufficiently large and $k \le n^{1-\eps}$, repeatedly asking random guesses until all but the secret code are excluded by the answers is an optimal Codebreaker strategy (up to constant factors). More specifically, using the probabilistic method and random guesses, he showed the existence of a \emph{deterministic non-adaptive} strategy for Codebreaker, that is,  a  set of $(2+\eps) n \frac{1+2 \log k}{\log (n/k)}$ guesses such that the answers uniquely determine any secret code Codemaker might have chosen (here and in the remainder, $\log n$ denotes the binary logarithm of $n$). These bounds hold even in the black-pegs only version of the game. 

For larger values of $k$, the situation is less understood. Note that the information-theoretic lower bound is $\Omega(n)$ for any number $k = n^\alpha$, $\alpha > 0$ a constant, of colors. 
For $k$ between $n$ and $n^2$, Chv\'atal presented a deterministic adaptive strategy using $2n \log k +4n$ guesses. For $k = n$, this strategy does not need white answer-pegs. 
Chv\'atal's result has been improved subsequently. 
\ifLongCites{\cite{Chen96}}{Chen, Cunha, and Homer~\cite{Chen96}} showed that for any $k \geq n$, 
$2n \lceil \log n \rceil+2n+\lceil k/n \rceil+2$ guesses suffice.
\ifLongCites{\cite{Goodrich09}}{Goodrich~\cite{Goodrich09}} proved an upper bound of
$n \lceil \log k \rceil +\lceil(2-1/k)n\rceil+k$ for the number of guesses needed to win the Mastermind game with an arbitrary number $k$ of colors and black answer-pegs only. This was again improved by \ifLongCites{\cite{JaegerP11}}{J\"ager and Peczarski~\cite{JaegerP11}}, who showed an upper bound of $n \lceil \log n \rceil - n + k + 1$ for the case $k > n$ and $n \lceil \log k \rceil + k$ for the case $k \le n$.
Note that for the case of $k = n$ colors and positions, all these results give the same asymptotic bound of $O(n\log n$).

\subsection{Our Contribution}

The results above show that Mastermind is well understood for $k \le n^{1-\eps}$, where we know the correct number of queries apart from constant factors. In addition, a simple non-adaptive guessing strategy suffices to find the secret code, namely casting random guesses until the code is determined by the answers.

On the other hand, for $k = n$ and larger, the situation is less clear. The best known upper bound, which is $O(n)$ (and tight) for $k = n^{\alpha}$, $0<\alpha<1$ a constant, suddenly increases to $O(n \log n)$ for $k = n$, while the information-theoretic lower bound remains at $\Omega(n)$. 

In this work, we prove that indeed there is a change of behavior around $k = n$. We show that, for $k = \Theta(n)$, the random guessing strategy, and, in fact, any other non-adaptive strategy, cannot find the secret code with an expected number of less than $\Theta(n \log n)$ guesses. This can be proven via an entropy compression argument as used by \ifLongCites{\cite{Moser09}}{Moser~\cite{Moser09}}, cf.\ Theorem~\ref{thm:nonadaptive}. For general $k$, our new lower bound for non-adaptive strategies is $\Omega(n \log(k) / \max\{\log(n/k), 1\})$. 
We also show that this lower bound is tight (up to constant factors). 
In fact, for $k\leq n$, $O(n \log(k) / \max\{\log(n/k), 1\})$ random guesses suffice to determine the secret code. That is, we extend Chv\'atal's result from $k\leq n^{1-\eps}$, $\eps > 0$ a constant, to all $k\leq n$. 


The main contribution of our work is a (necessarily adaptive) strategy that for $k=n$ finds the secret code with only $O(n \log\log n)$ queries.  
This reduces the $\Theta(\log n$) gap between the previous-best upper and the lower bound to $\Theta(\log\log n)$. 
Like the previous strategies for $k \le n$, our new one does not use white answer-pegs. Our strategy also improves the current best bounds for other values of $k$ in the vicinity of $n$; see Theorem~\ref{thm:main} below for the precise result.

The central part of our guessing strategy is setting up suitable \emph{coin-weighing problems}, solving them, and using the solution to rule out the possible occurrence of some colors at some positions. By a result of \ifLongCites{\cite{grebinski-journal}}{Grebinski and Kucherov~\cite{grebinski-journal}}, these coin weighing problems can be solved by relatively few independent random weighings. 

While our strategy thus is guided by probabilistic considerations, it can be derandomized to obtain a \emph{deterministic} $O(n\log\log n)$ strategy for black-peg Mastermind with $k=n$ colors. Moreover, appealing to an algorithmic result of \ifLongCites{\cite{Bshouty09}}{Bshouty~\cite{Bshouty09}} instead of Grebinski and Kucherov's result, we obtain a strategy that can be realized as a deterministic \emph{polynomial-time} codebreaking algorithm.

We also improve the current-best bounds for Mastermind with black and white answer-pegs, which stand at $O(n \log n)$ for $n\le k \le n^2 \log\log n$. For these $k$, we prove that $O(n \log\log n)$ guesses suffice. We point out that this improvement is not an immediate consequence of our $O(n \log\log n)$ bound for $k = n$ black-peg Mastermind. Reducing the number of colors from $k$ to $n$ is a non-trivial sub-problem as well. For example, when $k \ge n^{1 + \eps}$, Chv\'atal's strategy  for the game with black and white answer-pegs also uses $\Theta(n \log n)$ guesses to reduce the number of colors from $k$ to $n$, before employing a black-peg strategy to finally determine the secret code.

\subsection{Related Work}
\label{sec:review}

A number of results exist on the computational complexity of evaluating given guesses and answers. \ifLongCites{\cite{StuckmanZ06}}{Stuckman and Zhang~\cite{StuckmanZ06}} showed 
that it is NP-hard to decide whether there exists a secret code consistent with a given sequence of queries and black- and white-peg answers.
This result was extended to black-peg Mastermind by Goodrich~\cite{Goodrich09}.
More recently, \ifLongCites{\cite{Viglietta12}}{Viglietta~\cite{Viglietta12}} showed that both hardness results apply also to the setting with only $k=2$ colors. In addition, he proved that \emph{counting} the number of consistent secret codes is $\#$P-complete.

Another intensively studied question in the literature concerns the computation of (explicit) optimal winning strategies for small values of $n$ and $k$. 
As described above, the foundation for these works was laid by Knuth's famous paper~\cite{Knuth77} for the case with $n=4$ positions and $k=6$ colors. 
His strategy is worst-case optimal.
\ifLongCites{\cite{KoyamaL93}}{Koyama and Lai~\cite{KoyamaL93}} studied the average-case difficulty of Mastermind. They gave a strategy that solves Mastermind in an expected number of about 4.34 guesses if the secret string is sampled uniformly at random from all $6^4$ possibilities, and they showed that this strategy is optimal.
Today, a number of worst-case and average-case optimal winning strategies for different (small) values of $n$ and $k$ are known---both for the black- and white-peg version of the game~\cite{Goddard04,JaegerP09} and for the black-peg version~\cite{JaegerP11}. 
Non-adaptive strategies for specific values of $n$ and $k$ were studied in~\cite{Goddard03}.  

In the field of computational intelligence, Mastermind is used as a benchmark problem. For several heuristics, among them genetic and evolutionary algorithms, it has been studied how well they play Mastermind~\cite{KaliskerC03,KovacsT03,Berghman09,GuervosCM11,GuervosMC11}.

Trying to understand the intrinsic difficulty of a problem for such heuristics, \ifLongCites{\cite{DrosteJW06}}{Droste, Jansen, and Wegener~\cite{DrosteJW06}} suggested to use a query complexity variant (called black-box complexity). For the so-called onemax test-function class, an easy benchmark problem in the field of evolutionary computation, the black-box complexity problem is just the Mastermind problem for two colors. This inspired, among others, the result~\cite{DoerrW12} showing that a memory-restricted version of Mastermind (using only two rows of the board) can still be solved in $O(n / \log n)$ guesses when the number of colors is constant.
%


Several privacy problems have been modeled via the Mastermind game. \ifLongCites{\cite{Goodrich092}}{Goodrich~\cite{Goodrich092}} used black-peg Mastermind to study the extent of private genomic data leaked by comparing DNA-sequences (even when using protocols only revealing the degree of similarity). \ifLongCites{\cite{FocardiL10}}{Focardi and Luccio~\cite{FocardiL10}} showed that certain API-level attacks on user PIN data can be seen as an extended Mastermind game.

\subsection{Organization of this paper}

We describe and analyze our $O(n\log\log n)$ strategy for $k=n$ colors in Section~\ref{sec:nloglogn}. In Section~\ref{sec:black-and-white} we present a strong connection between the black-pegs only and the classic (black and white pegs) version of Mastermind. This yields, in particular, the claimed bound of $O(n\log\log n)$ for the classic version with $n\leq k \leq n^2 \log\log n$ colors. In Section~\ref{sec:nonadaptive} we analyze non-adaptive strategies. We prove a lower bound via entropy compression and show that it is tight for $k \le n$ by extending Chavatal's analysis of random guessing to all $k \le n$.

\section{The $O(\lowercase{n \log\log n})$ Adaptive Strategy}
\label{sec:nloglogn}


In this section we present the main contribution of this work, a black-pegs only strategy that solves Mastermind with $k=n$ colors in $O(n\log\log n)$ queries. We state our results for an arbitrary number $k=k(n)$ of colors; they improve upon the previously known bounds for all $k=o(n\log n)$ with $k\geq n^{1-\eps}$ for every fixed $\eps>0$.

\begin{theorem}
\label{thm:main} For Mastermind with $n$ positions and $k=k(n)$ colors, the following holds.
\begin{itemize}
\item If $k= \Omega(n)$ then there exists a randomized winning strategy that uses black pegs only and needs an expected number of $O(n \log \log n + k)$ guesses. 
\item If $k= o(n)$ then there exists a randomized winning strategy that uses black pegs only and needs an expected number of $O\left(n \log \left(\frac{\log n}{\log (n/k)}\right)\right)$ guesses. 
\end{itemize}
\end{theorem}
The $O$-notation in Theorem~\ref{thm:main} only hides absolute constants. Note that, setting $k=:n^{1-\delta}$, $\delta=\delta(n)$, the bound for $k=o(n)$ translates to $O(n\log(\delta^{-1}))$. 

We describe our strategy and prove Theorem~\ref{thm:main} in Sections~\ref{sec:ideas}-\ref{sec:analysis}. We discuss the derandomization of our strategy in Section~\ref{sec:derandomization}.


\subsection{Main Ideas} 
\label{sec:ideas}
Our goal in this section is to give an informal sketch of our main ideas, and to outline how the $O(n\log\log n)$ bound for $k=n$ arises. For the sake of clarity, we nevertheless present our ideas in the general setting---it will be useful to distinguish between $k$ and $n$ notationally. As justified in Section~\ref{sec:preliminaries} below, we assume that $k\leq n$ and that both $k$ and $n$ are powers of two.

A simple but crucial observation is that when we query a string $x \in [k]^n$ and the answer $\eq(z,x)$  is 0 (recall that $z$ denotes Codemaker's secret color code), then we know that all queried colors are wrong for their respective positions; i.e., we have $z_i\neq x_i$ for all $i \in [n]$. To make use of this observation, we maintain, for each position $i$, a set $C_i\subseteq[k]$ of colors that we still consider possible at position $i$. Throughout our strategy we reduce these sets successively, and once $|C_i|=1$ for all $i\in[n]$ we have identified the secret code~$z$. Variants of this idea have been used by several previous authors~\cite{Chvatal83,Goodrich09}.
 
Our strategy proceeds in phases. In each phase we reduce the size of all sets $C_i$ by a factor of two. Thus, before the $j$th phase we will have $|C_i|\leq k/2^{j-1}$ for all $i \in [n]$. 
Consider now the beginning of the $j$th phase, and assume that all sets $C_i$ have size exactly $k':=k/2^{j-1}$.  Imagine we query a random string $r$ sampled uniformly from $C_1\times \dots \times C_n$. The expected value of $\eq(z,r)$ is $n/k'$, and the probability that $\eq(z,r)=0$ is $(1-1/k')^n\leq e^{-n/k'}$. If $k'$ is significantly smaller than $n$, this probability is very small, and we will not see enough $0$-answers to exploit the simple observation we made above. However, if we group the $n$ positions into $m:=4n/k'$ blocks of equal size $k'/4$, the expected contribution of each such block is $1/4$, and the probability that a fixed such block contributes $0$ to $\eq(z,r)$ is $(1-1/k')^{k'/4}\approx e^{-1/4}$, i.e., constant. We will refer to blocks that contribute $0$ to $\eq(z,r)$ as \emph{$0$-blocks} in the following. For a random query we expect a constant fraction of all $m$ blocks to be $0$-blocks. If we can identify \emph{which} blocks these are, we can rule out a color at each position of each such block and make progress towards our goal. 

As it turns out, the identification of the $0$-blocks can be reduced to a \emph{coin-weighing problem} that has been studied by several authors; see~\cite{grebinski-journal,Bshouty09} and references therein. Specifically, we are given $m$ coins of unknown integer weights and a spring scale. We can use the spring scale to determine the total weight of an arbitrary subset of coins in one weighing. Our goal is to identify the weight of every coin with as few weighings as possible.

In our setup, the `coins' are the blocks we introduced above, and the `weight' of each block is its contribution to $\eq(z,r)$. To simulate weighings of subsets of coins by Mastermind queries, we use `dummy colors' for some positions, i.e., colors that we already know to be wrong at these positions. Using these, we can simulate the weighing of a subset of coins (=blocks) by copying the entries of the random query $r$ in blocks that correspond to coins we wish to include in our subset, and by using dummy colors for the entries of all other blocks. 

Note that the \emph{total weight} of our `coins' is $\eq(z,r)$. Typically this value will be close to its expectation $n/k'$, and therefore of the same order of magnitude as  the number of blocks $m$. It follows from a coin-weighing result by \ifLongCites{\cite{grebinski-journal}}{Grebinski and  Kucherov~\cite{grebinski-journal}} that $O(m/\log m)$ random queries (of the described block form, simulating the weighing of a random subset of coins) suffice to determine the contribution of each block to $\eq(z,r)$ with some positive probability. As observed before, typically a constant fraction of all blocks contribute $0$ to $\eq(z,r)$, and therefore we may exclude a color at a constant fraction of all $n$ positions at this point.

Repeating this procedure of querying a random string $r$ and using additional `random coin-weighing queries' to identify the $0$-blocks eventually reduces the sizes of the sets $C_i$ below $k'/2$, at which point the phase ends. In total this requires $\Theta(k')$ rounds in which everything works out as sketched, corresponding to a total number of $\Theta(k'\cdot (m/\log m)) = \Theta(n / \log (4n/k'))$ queries for the entire phase.

Summing over all phases, this suggests that for $k=n$ a total number of
$$\sum_{j=1}^{\log k} O\left(\frac{n}{\log (4n/\left(\frac{k}{2^{j-1}} \right))}\right)  \stackrel{k = n}= O(n) \sum_{j=1}^{\log n} \frac{1}{j+1} = O(n \log \log n)
$$
queries suffice to determine the secret code $z$, as claimed in Theorem~\ref{thm:main} for $k=n$.

We remark that our precise strategy, Algorithm~\ref{algo:mastermind}, slightly deviates from this description. This is due to a technical issue with our argument once the number $k'$ of remaining colors drops below $C\log n$ for some $C>0$. Specifically, beyond this point the error bound we derive for a fixed position is not strong enough to beat a union bound over all $n$ positions. To avoid this issue, we stop our color reduction scheme before $k'$ becomes that small (for simplicity as soon as $k'$ is less than $\sqrt{n}$), and solve the remaining Mastermind problem by asking random queries from the remaining set $C_1\times \dots \times C_n$, as originally proposed by \ifLongCites{\cite{Erd63}}{Erd\H os and R\'enyi~\cite{Erd63}} and \ifLongCites{\cite{Chvatal83}}{Chv\'atal~\cite{Chvatal83}}. 

\subsection{Precise Description of Codebreaker's Strategy} 
\label{sec:strategy-description}
 
\subsubsection{Assumptions on $n$ and $k$, Dummy Colors}
\label{sec:preliminaries}
Let us now give a precise description of our strategy. We begin by determining a dummy color for each position, i.e., a color that we know to be wrong at that particular position. For this we simply query  the $n+1$ many strings $(1,1, \dots, 1), (2,1,1, \dots, 1), \dots, (2, 2, \dots, 2)\in [k]^n$. Processing the answers to these queries in order, it is not hard to determine the location of all $1$'s and $2$'s in Codemaker's secret string $z$. In particular, this provides us with a dummy color for each position.


Next we argue that for the main part of our argument we may assume that $n$ and $k$ are powers of two. To see this for $n$, note that we can simply extend Codemaker's secret string in an arbitrary way such that its length is the smallest power of two larger than $n$, and pretend we are trying to determine this extended string. To get the answers to our queries in this extended setting, we just need to add the contribution of the self-made extension part (which we determine ourselves) to the answers Codemaker provides for the original string. As the extension changes $n$ at most by a factor of two, our claimed asymptotic bounds are unaffected by this.

To argue that we may also assume $k$ to be a power of two, we make use of the dummy colors we already determined for the original value of $k$.  Similar to the previous argument, we increase $k$ to the next power of two and consider the game with this larger number of colors. To get the answers to our queries in this extended setting from Codemaker (who still is in the original setting), it suffices to replace every occurrence of a color that is not in the original color set with the dummy color at the respective position.

We may and will also assume that $k\leq n$. If $k>n$ we can trivially reduce the number of colors to $n$ by making $k$ monochromatic queries. With this observation the first part of Theorem~\ref{thm:main} follows immediately from the $O(n\log\log n)$ bound we  prove for the case $k=n$. 

\subsubsection{Eliminating Colors with Coin-Weighing Queries}

With these technicalities out of the way, we can focus on the main part of our strategy. 
As sketched above, our strategy operates in \emph{phases}, where in the $j$th phase we reduce the sizes of the sets $C_i$ from $k/2^{j-1}$ to  $k/2^j$. 
For technical reasons, we do not allow the sizes of $C_i$ to drop below $k/2^j$ during phase $j$; i.e., once we have $|C_i|=k/2^j$ for some position $i\in[n]$, we no longer remove colors from $C_i$ at that position and ignore any information that would allow us to do so.

Each phase is divided into a large number of \emph{rounds}, where a round consists of querying a random string $r$ and subsequently identifying the $0$-blocks (blocks that contribute $0$ to $\eq(z,r)$) by the coin-weighing argument outlined above.

To simplify the analysis, the random string $r$ is sampled from the same distribution throughout the entire phase. Specifically, at the beginning of phase $j$ we define the set $\cR_j:=C_1\times \dots \times C_n$, and sample the random string $r$ uniformly at random from $\cR_j$ in each round of phase $j$. Note that we do not adjust $\cR_j$ during phase $j$; information about excluded colors we gain during phase $j$ will only be used in the definition of the set $\cR_{j+1}$ in phase $j+1$. 

We now introduce the formal setup for the coin-weighing argument. As before we let $k'~:=~k/2^{j-1}$ and partition the $n$ positions into $m := 4 n / k'$ blocks of size $k'/4$. 
More formally, for every $s \in [m]$ we let $B_s := \{ (s-1) k'/4 +1, \dots, s k'/4 \}$ denote the indices of block $s$, and we denote by $v_s: = | \{ i \in B_s : z_i = r_i \} |$ the contribution of block $B_s$ to $\eq(z,r)$. (Note that $\sum_{s \in [m]} v_s = \eq(z,r)$.) As indicated above we wish to identify the 0-blocks, that is, the indices $s \in [m]$ for which $v_s = 0$.

For $y \in \{0,1\}^m$, define $r_y$ as the query that is identical to $r$ on the blocks $B_s$ for which $y_s = 1$, and identical to the string of dummy colors on all other blocks. Thus $\eq(z,r_y) = \sum_{s \in [m], y_s = 1} v_s$. With this observation, identifying the values $v_s$ from a set of queries of form $r_y$ is equivalent to a coin-weighing problem in which we have $m$ coins with positive integer weights that sum up to $\eq(z,r)$: Querying $r_y$ in the Mastermind game provides exactly the information we obtain from weighing the set of coins indicated by $y$.

We will only bother with the coin-weighing if the initial random query of the round satisfies $\eq(z,r)\leq m/2$. (Recall that the expected value of $\eq(z,r)$ is $m/4$.) If this is the case, we query an appropriate number $f(m)$ of strings of form $r_y$, with $y\in\{0,1\}^m$ sampled uniformly at random (u.a.r.) and independently. The function $f(m)$ is implicit in the proof of the coin-weighing result of~\cite{grebinski-journal}; it is in $\Theta(m/\log m)$ and guarantees that the coin-weighing succeeds with probability at least $1/2$. Thus with probability at least $1/2$, these queries determine all values $v_s$ and, in particular, identify all $0$-blocks. Note that the inequality $\eq(z,r)\leq m/2$ also guarantees that at least half of the $m$ blocks are $0$-blocks.

We say that a round is \emph{successful} if $\eq(z,r) \leq m/2$ and if the coin-weighing successfully identifies all $0$-blocks. In each successful round, we update the sets $C_i$ as outlined above; i.e., for each position $i$ that is in a $0$-block and for which $|C_i|>k'/2$ we set $C_i:=C_i\setminus\{r_i\}$. Note that it might happen that $r_i$ is a color that was already removed from $C_i$ in an earlier round of the current phase, in which case $C_i$ remains unchanged.
If a round is unsuccessful we do nothing and continue with the next round.

This completes the description of our strategy for a given phase. We abandon this color reduction scheme once $k'$ is less than $\sqrt{n}$. At this point, we simply ask queries sampled uniformly and independently at random from the current set $\cR=C_1 \times \dots \times C_n$. We do so until the answers uniquely determine the secret code $z$. 
It follows from \ifLongCites{\cite{Chvatal83}}{Chv\'atal's result~\cite{Chvatal83}} that the expected number of queries needed for this is $O(n \log k' /  \log (n/k'))= O(n)$.

This concludes the description of our strategy. It is summarized in Algorithm~\ref{algo:mastermind}. Correctness is immediate from our discussion, and it remains to bound the expected number of queries the strategy makes.

\begin{algorithm2e} 
\label{algo:mastermind}
Determine a dummy color for each position\label{line:dummy}\;
\lForEach{$i\in[n]$}{$C_i\assign[k]$\;} 
$j \assign 0$ and $k'\assign k$\;
\While{$k' > \sqrt{n}$}{
  $j \assign j+1$, 
  $k' \assign k/2^{j-1}$, 
  $\cR_j \assign C_1\times \dots \times C_n$,
   and $m \assign 4n/k'$\;
   \Repeat{$\forall i\in [n]: |C_i|= k'/2$}{
      Select a string $r$ u.a.r.\ from $\cR_j$ and query $\eq(z,r)$\label{line:sample}\;
      \If{$\eq(z,r)\leq m/2$ \label{Bedingung}}{
        \For{$i=1, \ldots, f(m)$ \quad /* $f(m) = \Theta(m/\log m)$ /* \label{line:coins1}}
           {Sample $y$ u.a.r.\ from $\{0,1\}^m$  and query $\eq(z,r_y)$\label{line:coins2}\;}
          \If{\normalfont{these $f(m)$ queries determine the $0$-blocks of $r$}}{
              \ForEach{$i\in[n]$}
                 {\lIf{\normalfont{$i$ is in a $0$-block and $|C_i|> k'/2$}}
                    {$C_i \assign C_i\setminus\{r_i\}$\;}  }
             }
      }  
    }
}
$\cR \assign C_1\times \dots \times C_n$\;
Select strings $r$ independently and u.a.r.\ from $\cR$ and query $\eq(z,r)$ until $z$ is determined\;
\caption{Playing Mastermind with many colors}
\end{algorithm2e}

\subsection{Proof of Theorem~\ref{thm:main}}
\label{sec:analysis}

We begin by bounding the expected number of rounds in the $j$th phase.

\begin{claim} \label{claim:one-phase}
 The expected number of rounds required to complete phase $j$ is $O(k')=O(k/2^j)$.
\end{claim}

\begin{proof}
We first show that a round is successful with probability at least $1/4$. Recall that $\eq(z,r)$ has an expected value of $n/k'=m/4$. 
Thus, by Markov's inequality, we have $\eq(z,r)\leq m/2$ with probability at least $1/2$. Moreover, as already mentioned, the proof of the coin-weighing result by 
\ifLongCites{\cite{grebinski-journal}}{Grebinski and Kucherov~\cite{grebinski-journal}} 
implies that our $f(m)=\Theta(m/\log m)$ random coin-weighing queries identify all $0$-blocks with probability at least $1/2$. 
Thus, in total the probability for a successful round is at least $1/2\cdot 1/2=1/4$.

We continue by showing that the probability that a successful round decreases the number of available colors for a fixed position, say position 1, is at least $1/4$. Note that this happens if $r \in \cR_j$ satisfies the following two conditions:
\begin{itemize}
\item[(i)] $v_1=0$, i.e., block $B_1$ is a 0-block with respect to $r$, and
\item[(ii)] $r_1\in C_1$, i.e., the color $r_1$ has not been excluded from $C_1$ in a previous round of phase $j$.
\end{itemize}
 For (i) recall that in a successful round at least $m/2$ of the $m$ blocks are $0$-blocks. It follows by symmetry that $B_1$ is a 0-block with probability at least $1/2$. Moreover, conditional on (i), $r_1$ is sampled uniformly at random from the $k'-1$ colors that are different from $z_1$ and were in $C_1$ at the beginning of the round. Thus the probability that $r_1$ is in the current set $C_1$ is $|C_1|/(k'-1)$, which is at least $1/2$ because we do not allow $|C_1|$ to drop below $k'/2$. We conclude that, conditional on a successful round, the random query $r$ decreases $|C_1|$ by one with probability at least $1/2\cdot 1/2=1/4$.

Thus, in total, the probability that a round decreases $|C_1|$ by one is at least $1/4\cdot 1/4=1/16$ throughout our strategy. It follows that the probability that after $t$ successful rounds in phase $j$ we 
still have $|C_1|>k'/2$ is bounded by the probability that in $t$ independent Bernoulli trials with  success probability $1/16$ we observe fewer than $k'/2$ successes. If $t/16\geq k'$, by Chernoff bounds this probability is bounded by $e^{-ct}$ for some absolute constant $c>0$. 

Let us now denote the number of rounds phase $j$ takes by the random variable $T$. By a union bound, the probability that $T\geq t$, i.e., that after $t$ steps at one of the positions $i\in[n]$ we still have $|C_i|>k'/2$, is bounded by $ne^{-ct}$ for $t\geq 16k'$. It follows that \begin{equation}
\mathbb E[T] =  \sum_{t \ge 1} \Pr[ T \ge t ] 
\le  16k' + n  \sum_{t > 16k'} e^{-ct} =  16k' + n  e^{-\Omega(k')}= O(k'),
\end{equation}
where the last step is due to $k'\geq \sqrt{n}=\omega(\log n)$.
\end{proof}

With Claim~\ref{claim:one-phase} in hand, we can bound the total number of queries  required throughout our strategy by a straightforward calculation.

\begin{proof}[Proof of Theorem~\ref{thm:main}]
Recall that for each phase $j$ we have $m=\Theta(n/k')=\Theta(n/(k/2^{j-1}))$ and that $f(m)=\Theta(m/\log m)$. Thus by Claim~\ref{claim:one-phase}, the expected number of queries our strategy asks in phase $j$ is bounded by $$O(k')\cdot (1 + f(m)) = O\left(\frac{n}{\log(\frac{n}{k/2^j})}\right) = O\left( \frac{n}{\log(n/k) +j}\right).$$

It follows that throughout the main part of our strategy we ask an expected number of queries of at most 
$$O(n) \sum_{j=1}^{\log k} \frac{1}{\log(n/k) +j} = O\big(n \big(\log \log n - \log\log (n/k)\big)\big)= O\left(n \log \left(\frac{\log n}{\log (n/k)}\right)\right).
$$
(This calculation is for $k<n$; as observed before, for $k=n$ a very similar calculation yields a bound of $O(n\log\log n)$.)
As the number of queries for determining the dummy colors and for wrapping up at the end  is only $O(n)$, Theorem~\ref{thm:main} follows.
\end{proof}

\subsection{Derandomization}
\label{sec:derandomization}

The strategy we presented in the previous section can be derandomized and implemented as a polynomial-time algorithm.

\begin{theorem}
\label{thm:derandupper}
The bounds stated in Theorem~\ref{thm:main} can be achieved by a deterministic winning strategy. 
Furthermore, this winning strategy can be realized in polynomial time.
\end{theorem}

\begin{proof}
The main loop of the algorithm described above uses randomization in two places: for generating the random string $r$ of each round (line~\ref{line:sample} in Algorithm 1), and for generating the $f(m)$ many random coin-weighing queries $r_y$ used to identify the $0$-blocks of $r$ if $\eq(z,r)\leq m/2$ (line~\ref{line:coins2}).

The derandomization of the coin-weighing algorithm is already given in the work of \ifLongCites{\cite{grebinski-journal}}{Grebinski and Kucherov~\cite{grebinski-journal}}. They showed that a set of $f'(m)= \Theta(m/\log m)$ random coin-weighing queries $y^1, \ldots, y^{f'(m)}$, sampled from $\{0,1\}^m$ independently and uniformly at random, has, with some positive probability, the property that it distinguishes any two distinct coin-weighing instances in the following sense: For any two distinct vectors $v, w$ with non-negative integer entries such that $\sum_{s\in[m]}{v_s} \leq m/2$ and $\sum_{s\in[m]}{w_s} \leq m/2$, there exists an index $j\in[f'(m)]$ for which $\sum_{s\in[m],y^j_s=1}{v_s}\neq\sum_{s\in[m],y^j_s=1}{w_s}$.
It follows by the probabilistic method that, deterministically, there is a set $D\subseteq \{0,1\}^m$ of size at most $f'(m)$ such that the answers to the corresponding coin-weighing queries identify every possible coin-weighing instance. Hence we can replace the $f(m)$ random coin-weighing queries of each round by the $f'(m)$ coin-weighing queries corresponding to the fixed set $D$.

It remains to derandomize the choice of $r$ in each round. As before we consider $m:=4n/k'$ blocks of size $k'/4$, where $k'$ is the size of the sets $C_i$ at the beginning of a phase. To make sure that a constant fraction of all queries in a phase satisfy $\eq(z,r)\leq m/2$ (compare line~\ref{Bedingung} of Algorithm~\ref{algo:mastermind}), we ask a set of $k'$ queries such that, for each position $i\in[n]$, every color in $C_i$ is used at position $i$ in exactly one of these queries. (If all sets $C_i$ are equal, this can be achieved by simply asking $k'$ monochromatic queries.) The sum of all returned scores must be exactly $n$, and therefore we cannot get a score of more than $m/2 = 2n/k'$ for more than $k'/2$ queries. In this way we ensure that for at least $k-k'/2=k'/2$ queries we get a score of at most $m/2$.

As in the randomized version of our strategy, in each of these $k'/2$ queries at least half of the blocks must be $0$-blocks. We can identify those by the derandomized coin-weighing discussed above. Consider now a fixed block. As it has size $k'/4$, it can be a \emph{non}-0-block in at most $k'/4$ queries. Thus it is a $0$-block in at least $k'/2-k'/4=k'/4$ of the queries.

To summarize, we have shown that by asking $k'$ queries of the above form we get at least $k'/2$ queries of score at most $m/2$. For each of them we identify the $0$-blocks by coin-weighing queries. This allows us to exclude at least $k'/4$ colors at each position. I.e., as in the randomized version of our strategy we can reduce the number of colors by a constant factor using only $O(k'\cdot m/\log m)=O(n/\log(4n/k))$ queries. By similar calculations as before, the same asymptotic bounds follow. 

We abandon the color reduction scheme when $k'$ is a constant. At this point, we can solve the remaining problem in time $O(n)$ by repeatedly using the argument we used to determine the dummy colors in Section~\ref{sec:preliminaries}.

Note that all of the above can easily be implemented in polynomial time if we can solve the coin-weighing subproblems in polynomial time. An algorithm for doing the latter is given in the work of \ifLongCites{\cite{Bshouty09}}{Bshouty~\cite{Bshouty09}}. Using this algorithm as a building block, we obtain a deterministic polynomial-time strategy for Codebreaker that achieves the bounds stated in Theorem~\ref{thm:main}.
\end{proof}

\section{Mastermind with Black and White Answer-Pegs} \label{sec:black-and-white}


In this section, we analyze the Mastermind game in the classic version with both black and white answer-pegs. Interestingly, there is a strong general connection between the two versions. Roughly speaking, we can use a strategy for the $k = n$ black-peg game to learn which colors actually occur in the secret code of a black/white-peg game with $n$ positions and $n^2$ colors. Having thus reduced the number of relevant colors to at most $n$, Codebreaker can again use a $k = n$ black-peg strategy (ignoring the white answer-pegs) to finally determine the secret code. 

More precisely, for all $k, n \in \N$ let us denote by $b(n,k)$ the minimum (taken over all strategies) maximum (taken over all secret codes) expected number of queries needed to find the secret code in a black-peg Mastermind game with $k$ colors and $n$ positions. Similarly, denote by $bw(n,k)$ the corresponding number for the game with black and white answer-pegs. Then we show the following.

\begin{theorem}\label{thmbw}
  For all $k, n \in \N$ with $k \ge n$, \[bw(n,k) = \Theta(k/n + b(n,n)).\]
\end{theorem}

Combining this with Theorem~\ref{thm:main}, we obtain a bound of $O(n\log\log n)$ for black/white Mastermind with $n\leq k \leq n^2\log\log n$ colors, improving all previous bounds in that range.

For the case $k \leq n$ it is not hard to see that $bw(n,k) = \Theta(b(n,k))$, see Corollary~\ref{cor:bw} below. Together with Theorem~\ref{thmbw}, this shows that to understand black/white-peg Mastermind for all $n$ and $k$, it suffices to understand black-peg Mastermind for all $n$ and $k$.

Before proving Theorem~\ref{thmbw}, let us derive a few simple preliminary results on the relation of the two versions of the game.

\begin{lemma}\label{lem:bw}
  For all $n, k$, \[bw(n,k) \ge b(n,k) - k + 1.\]
\end{lemma}

\begin{proof}
We show that we can simulate a strategy in the black/white Mastermind game by one receiving only black-pegs answers and using $k-1$ more guesses. Fix a strategy for black/white Mastermind. Our black-peg strategy first asks $k-1$ monochromatic queries. This tells us how often each of the $k$ color arises in the secret code. From now on, we can play the strategy for the black/white game. While we only receive black answer-pegs, we can compute the number of white pegs we would have gotten in the black/white game from the just obtained information on how often each color occurs in the code. With this information available, we can indeed play as in the given strategy for black/white Mastermind.
\end{proof}

Lemma~\ref{lem:bw} will be used to prove that the $b(n,n)$ term in the statement of Theorem~\ref{thmbw} cannot be avoided. As a corollary, it yields that white answer-pegs are not extremely helpful when $k = O(n)$.

\begin{corollary}\label{cor:bw}
  For all $k \le n$, \[bw(n,k) = \Theta(b(n,k)).\]
\end{corollary}

\begin{proof}
  Obviously, $bw(n,k) \le b(n,k)$ for all $n,k$. If $k = o(n)$, then the information theoretic lower bound $b(n,k) = \Omega(n \log k / \log n)$ is of larger order than $k$, hence the lemma above shows the claim. For $k = \Theta(n)$, note first that both $b(n,k)$ and $bw(n,k)$ are in $\Omega(n)$ due to the information theoretic argument. If $b(n,k) = O(n)$, there is nothing to show. If $b(n,k) = \omega(n)$, we again invoke Lemma~\ref{lem:bw}.
\end{proof}

In the remainder of this section, we prove Theorem~\ref{thmbw}. 
To describe the upper bound, let us fix the following notation. Let $C$ be the set of all available colors and $k = |C|$. Denote by $z \in C^n$ the secret code chosen by Codemaker. Denote by $C^* := \{z_i \mid i \in [n]\}$ the (unknown) set of colors in $z$.

Codebreaker's strategy leading to the bound of Theorem~\ref{thmbw} consists of roughly these three steps. 

(1) Codebreaker first asks roughly $k/n$ guesses containing all colors. Only colors in a guess receiving a positive answer can be part of the secret code, so this reduces the number of colors to be regarded to at most $n^2$. Also, Codebreaker can learn from the answers the cardinality $n'$ of $C^*$, that is, the number of distinct colors in the secret code.

(2)
By asking an expected number of $\Theta(n')$ (dependent) random queries, Codebreaker learns $n'$ disjoint sets of colors of size at most $n$ such that each color of $C^*$ is contained in exactly one of these sets. Denote by $k'$ the cardinality of a largest of these sets. 

(3) Given such a family of sets, Codebreaker can learn $C^*$ with an expected number of $b(n', k')$ queries by simulating an optimal black-peg Mastermind strategy. Once $C^*$ is known, an expected number of $b(n,n')$ queries determine the secret code, using an optimal black-peg strategy for $n'$ colors.

Each of these steps is made precise in the following. Before doing so, we remark that after a single query Codebreaker may detect $|C^* \cap X|$ for any set $X$ of at most $n$ colors via a single Mastermind query to be answered by black and white answer-pegs.

\begin{lemma}\label{lemfirst}
  For an arbitrary set $X$ of at most $n$ colors, let $\col(X) := |C^* \cap X|$, the number of colors of $X$ occurring in the secret code. After a single initial query, Codebreaker can learn $\col(X)$ for any $X$ via a single Mastermind query to be answered by black and white pegs.
\end{lemma}

\begin{proof}
  As the single initial query, Codebreaker may ask $(1,\ldots,1)$, the code consisting of color~1 only. Denote by $b$ the number of black pegs received (there cannot be a white answer-peg). This is the number of occurrences of color~$1$ in the secret code. 
  
  Let $X \subseteq C$, $\nu := |X| \le n$. To learn $\col(X)$, Codebreaker extends $X$ to a multiset of $n$ colors by adding the color 1 exactly $n-\nu$ times and guesses a code arbitrarily composed of this multiset of colors. Let $y$ be the total number of (black and white) answer-pegs received. Then $\col(X) = y - \min\{n-\nu,b\}$, if $1 \notin X$ or $b = 0$, and $\col(X) = y - \min\{n-\nu,b-1\}$ otherwise. 
\end{proof}

To ease the language, we shall call a query determining $\col(X)$ a \emph{color query}. We now show that using roughly $k/n$ color queries, Codebreaker can learn the number $|C^*|$ of different colors occurring in the secret code and exclude all but $n |C^*|$ colors.
 
\begin{lemma}\label{lemreduce}
  With  $\lceil k/n \rceil$ color queries, Codebreaker can learn both $|C^*|$ and a superset $C_0$ of $C^*$ consisting of at most $n |C^*|$ colors.
\end{lemma}

\begin{proof}
   Let $X_1, \ldots, X_{\lceil k/n \rceil}$ be a partition of $C$ into sets of cardinality at most $n$. By asking the corresponding  $\lceil k/n \rceil$ color queries, Codebreaker immediately learns $|C^*| := \sum_{i = 1}^{\lceil k/n \rceil} \col(X_i)$. Also, $C_0 := \bigcup \{X_i \mid \col(X_i) > 0\}$ is the desired superset.
\end{proof}

\begin{lemma}\label{lempartition}
  Assume that Codebreaker knows the number $n' = |C^*|$ of different colors in $z$ as well as a set $C_0 \supseteq C^*$ of colors such that $|C_0| \le n |C^*|$.
  
  Then with an expected number of $\Theta(n')$ color queries, Codebreaker can find a family $C_1, \ldots, C_{n'}$ of disjoint subsets of $C_0$, each of size at most $ \lceil |C_0|/n' \rceil \le n$, such that $C^* \subseteq C_1 \cup \ldots \cup C_{n'}$ and $|C^* \cap C_i| = 1$ for all $i \in [n']$. 
\end{lemma}

\begin{proof}
  Roughly speaking, Codebreaker's strategy is to ask color queries having an expected answer of one. With constant probability, such a query contains exactly one color from $C^*$. Below is a precise formulation of this strategy.
  
\begin{algorithm2e} 
\label{algo:blackwhite}
 \While{$n' > 0$}{
    $k' \assign \lceil |C_0|/n' \rceil$\;
    Let $C_{n'}$ be a random subset of $C$ with $|C_{n'}| = k'$\; 
    Ask the color query $C_{n'}$\;
    \If{$\col(C_{n'}) = 1$}{
       $C_0 \assign C_0 \setminus C_{n'}$\;
       $n' \assign n' - 1$\;       
       }
     }
\caption{Codebreaker's strategy}
\end{algorithm2e}

For the analysis, note first that the value of $k'$ during the application of the above strategy does not increase. In particular, all sets $C_i$ defined and queried have cardinality at most $\lceil |C_0|/n' \rceil \le n$. 
It is also clear that the above strategy constructs a sequence of disjoint $C_i$ and that for each color occurring in $z$ there is exactly one $C_i$ containing this color. 

It remains to prove the estimate on the expected number of queries. To this aim, we first note that throughout a run of this strategy, $n'$ is the number of colors of $C^*$ left in $C_0$. Hence the event ``$\col(C_{n'}) = 1$'' occurs with probability
\begin{align*}
\frac{n' k' (|C_0| - n')  \ldots (|C_0| - n' - k' + 2)}{|C_0| \ldots (|C_0| - k' + 1)} 
&\ge \frac{(|C_0| - n') \ldots (|C_0| - n' - k' + 2)}{(|C_0|-1) \ldots (|C_0| - k' + 1)} \\
& \ge \bigg( \frac{|C_0| - n' - k' + 2}{|C_0| - k' + 1} \bigg)^{k'-1} = \bigg( 1 - \frac{n' - 1}{|C_0| - k' + 1} \bigg)^{k'-1}\\
& \ge \bigg( 1 - \frac{n' - 1}{|C_0| - (|C_0|/n')} \bigg)^{k'-1}\\
& \ge \bigg( 1 - \frac{|C_0| / (k'-1)}{|C_0| - (|C_0|/n')} \bigg)^{k'-1}\\
& \ge \bigg( 1 - \frac{1}{(k'-1)(1 - 1/n')} \bigg)^{k'-1},\\
\end{align*}
which is bounded from below by a constant (the later estimates assume $n' \ge 2$; for $n'=1$ the second term of the sequence of inequalities already is one).

Consequently, with constant probability the randomly chosen $C_{n'}$ satisfies ``$\col(C_{n'}) = 1$''. Hence after an expected constant number of iterations of the while-loop, such a $C_{n'}$ will be found. Since each such success reduces the value of $n'$ by one, a total expected number of $\Theta(|C^*|)$ iterations suffices to find the desired family of sets $(C_i)_{i \in [n']}$.
\end{proof}

Given a family of sets as just constructed, Codebreaker can simulate a black-peg strategy to determine $C^*$.

\begin{lemma}\label{lemlast}
  Let $C_1, \ldots, C_{n'}$ be a family of disjoint subsets of $C$ such that $C^* \subseteq C_1 \cup \ldots \cup C_{n'}$ and $|C^* \cap C_i| = 1$ for all $i \in [n']$. Assume that $k' := \max\{|C_i| \mid i \in [n']\} \le n$. Then Codebreaker can detect $C^*$ using an expected number of $b(n', k')$ color queries.
\end{lemma}

\begin{proof}
  Let $z' \in C_1 \times \ldots \times C_{n'}$ be the unique such string consisting of colors in $C^*$ only. Note that in black-peg Mastermind, the particular sets of colors used at each position are irrelevant. Hence there is a strategy for Codebreaker to detect $z'$ using an expected number of $b(n',k')$ guesses from $C_1 \times \ldots \times C_{n'}$ and receiving black-peg answers only. 
  
  We now show that for each such query, there is a corresponding color query in the $(n,k)$ black/white Mastermind game giving the same answer. Hence we may simulate the black-peg game searching for $z'$ by such color queries. Since $z'$ contains all colors of $C^*$ and no other colors, once found, it reveals the set of colors occurring in the original secret code $z$.
  
  Let $y' \in C_1 \times \ldots \times C_{n'}$ be a query in the black-peg Mastermind game searching for $z'$. For each position $i\in[n']$, we have $z'_i=y'_i$ if and only if $y_i'\in C_i$ is the unique color from $C_i$ that is in $C^*$. As moreover the sets $(C_i)_{i\in [n']}$ are disjoint, we have $\eq(z',y') = \col(\{y'_1, \ldots, y'_{n'}\})$, and we can obtain this value (i.e., the black-peg answer for the guess $y'$ relative to $z'$)  by a color query relative to $z$.
  \end{proof}

%

Note that if our only goal is to find out $C^*$, then for $k \ll n^2$ we can be more efficient by asking more color queries in Lemma~\ref{lemreduce}, leading to a smaller set $C_0$, to smaller sets $C_i$ in Lemma~\ref{lempartition}, and thus to a smaller $k'$ value in Lemma~\ref{lemlast}. Since this will not affect the asymptotic bound for the total numbers of queries used in the black/white-peg game, we omit the details.

\begin{proof}[Proof of Theorem~\ref{thmbw}]
The upper bound follows easily from applying Lemmas~\ref{lemfirst} to~\ref{lemlast}, which show that Codebreaker can detect the set $C^*$ of colors arising in the secret code $z$ with an expected number of $1 + \lceil k/n \rceil + O(n) + b(n,n)$ guesses. Since $|C^*| \le n$, he can now use a strategy for black-peg Mastermind and determine $z$ with another expected number of $b(n,n)$ guesses. Note that $b(n,n) = \Omega(n)$, so this proves the upper bound.

We argue that this upper bound is optimal apart from constant factors. Assume first that the secret code is a random monochromatic string (Codemaker may even announce this). Fix a (possibly randomized) strategy for Codebreaker. With probability at least 1/2, this strategy does not use the particular color in any of the first $k/(2n)$ guesses. It then also did not guess the correct code. Hence the expected number of queries necessary to find the code is at least $k/(4n)$. 

We finally show that for $k \ge n$, also the $b(n,n)$ term cannot be avoided. By the information theoretic argument, there is nothing to show if $b(n,n) = \Theta(n)$. Hence assume $b(n,n) = \omega(n)$. We will show $bw(n,k) + n + 1 \ge bw(n,n)$. The claim then follows from $bw(n,n) = \Theta(b(n,n))$ (Corollary~\ref{cor:bw}).

We show that we can solve the $k=n$ color Mastermind game by asking $n+1$ preliminary queries and then simulating a strategy for black/white Mastermind with $n$ positions and $k > n$. As in Section~\ref{sec:preliminaries}, we use $n+1$ queries to learn for each position whether it has color $1$ or not. We then simulate a given strategy for $k > n$ colors as follows. In a $k$-color query, replace all colors greater than $n$ by color $1$. Since we know the positions of the pegs in color $1$, we can reduce the answers by the contribution of these additional 1-pegs in the query. This gives the answer we would have gotten in reply to the original query (since the secret code does not contain colors higher than $n$). Consequently, we can now simulate the $k$-color strategy in an $n$-color Mastermind game.
\end{proof}

\section{Non-Adaptive Strategies}
\label{sec:nonadaptive}

When analyzing the performance of non-adaptive strategies, it is not very meaningful to ask for the number of queries needed until the secret code is queried for the first time. 
Instead we ask for the number of queries needed to \emph{identify} it.

In their work on the $2$-color black-peg version of Mastermind, \ifLongCites{\cite{Erd63}}{Erd{\H{o}}s and R\'enyi~\cite{Erd63}} showed that random guessing needs, with high probability, $(2+o(1)) n /\log n$ queries to identify the secret code, and that this is in fact best possible among \emph{non-adaptive} winning strategies. The upper bound was derandomized by \ifLongCites{\cite{Lindstroem64,Lindstroem65}}{Lindstr\"om~\cite{Lindstroem64,Lindstroem65}} and, independently, by \ifLongCites{\cite{CantorM66}}{Cantor and Mills~\cite{CantorM66}}. That is, for 2-color black-pegs Mastermind a deterministic non-adaptive winning strategy using $(2+o(1)) n /\log n$ guesses exists, and no non-adaptive strategy can do better. 

For adaptive strategies, only a weaker lower bound of $(1+o(1))n /\log n$ is known. This bound results from the information-theoretic argument mentioned in Section~\ref{sec:previous}. 
It remains a major open problem whether there exists an adaptive strategy that achieves this bound. In fact, it is not even known whether adaptive strategies can outperform the random guessing strategy by \emph{any} constant factor. 

Here in this section we prove that for Mastermind with $k=\Theta(n)$ colors, adaptive strategies are indeed more powerful than non-adaptive ones, and outperform them even in order of magnitude. 
More precisely, we show that any non-adaptive strategy needs $\Omega(n \log n)$ guesses. Since we know from Section~\ref{sec:nloglogn} that adaptively we can achieve a bound of $O(n \log \log n)$, this separates the performance of non-adaptive strategies from that of adaptive ones. Our result answers a question left open in~\cite{Goddard03}.

The $\Omega(n\log n)$ bound for non-adaptive strategies is tight. As we will show in Theorem~\ref{thm:non-adaptive-lb} below, there exists a deterministic non-adaptive strategy that achieves the bound up to constant factors.

\subsection{Lower Bound for Non-Adaptive Strategies}

For the formal statement of the bound, we use the following notation.
A deterministic non-adaptive strategy is a fixed ordering $x^1, x^2, \ldots, x^{k^n}$ of all possible guesses, i.e., the elements of $[k]^n$. A randomized non-adaptive strategy is a probability distribution over such orderings. For a given secret code $z\in[k]^n$, we ask for the smallest index $j$ such that the queries $x^1, \ldots, x^j$ together with their answers $\eq(z,x^1), \ldots, \eq(z,x^j)$ uniquely determine $z$. 

Mastermind with non-adaptive strategies is also referred to as \emph{static} Mastermind~\cite{Goddard03}. 

\begin{theorem} 
\label{thm:nonadaptive}
For any (randomized or deterministic) non-adaptive strategy for black-peg Mastermind with $n$ positions and $k$ colors, the expected number of queries needed to determine a secret code $z$ sampled uniformly at random from $[k]^n$ is
$\Omega \left( \frac{n \log k}{\max\{\log(n/k), 1\}} \right)$.
\end{theorem}

Theorem~\ref{thm:nonadaptive} shows, in particular, that for any non-adaptive strategy there exists a secret code $z\in[k]^n$ which can only be identified after $\Omega\left(n \log k / \max\{\log(n/k), 1\} \right)$ queries. For $k\geq n$, this is an improvement of $\Theta(\log n)$ over the information-theoretic lower bound mentioned in the introduction. 
For the case $k=\Theta(n)$ Theorem~\ref{thm:nonadaptive} gives a lower bound of $\Omega(n \log n)$ guesses for every non-adaptive strategy, showing that adaptive strategies are indeed more powerful than non-adaptive ones in this regime (recall Theorem~\ref{thm:main}). 

To give an intuition for the correctness of Theorem~\ref{thm:nonadaptive}, note that for a uniformly chosen secret code $z \in [k]^n$, for any single fixed guess $x$ of a non-adaptive strategy the answer $\eq(z,x)$ is binomially distributed with parameters $n$ and $1/k$. 
That is, $\eq(z,x)$ will typically be within the interval $n/k \pm O(\sqrt{n/k})$.
Hence, we can typically encode the answer using $\log( O(\sqrt{n/k})) = O(\log(n/k))$ bits. Or, stated differently, our `information gain' is usually $O(\log(n/k))$ bits. Since the secret code `holds $n \log k$ bits of information', we would expect that we have to make $\Omega(n \log k / \log(n/k))$ guesses.

To turn this intuition into a formal proof, we recall the notion of entropy: 
For a discrete random variable $Z$ over a domain $D$, the \emph{entropy} of $Z$ is defined by
$H(Z) := - \sum_{z \in D} \Pr[Z = z] \log(\Pr[Z=z])$. 
Intuitively speaking, the entropy measures the amount of information that the random variable $Z$ carries. 
If $Z$ for example corresponds to a random coin toss with $\Pr[\text{`heads'}] = \Pr[\text{`tails'}] = 1/2$, then $Z$ carries 1 bit of information. 
However, a biased coin toss with $\Pr[\text{`heads'}] = 2/3$ carries less (roughly 0.918 bits of) information since we know that the outcome of heads is more likely.
In our proof we use the following properties of the entropy, which can easily be seen to hold for any two random variables $Z, Y$ over domains $D_Z, D_Y$.
\begin{itemize}
\item[(E1)] If $Z$ is determined by the outcome of $Y$, i.e., $Z = f(Y)$ for a deterministic function $f$, then we have
$H(Z) \le H(Y)$.
\item[(E2)] We have $H((Z, Y)) \le H(Z) + H(Y)$.
\end{itemize}
The inequality in (E2) holds with equality if and only if the two variables $Z$ and $Y$ are independent.

\begin{proof}[Proof of Theorem~\ref{thm:nonadaptive}] 
Below we show that there a time $s = \Omega \left( \frac{n \log k}{\max\{\log(n/k), 1\}} \right)$ such that
any deterministic strategy at any time earlier that $s$ determines less than half of the secret codes. Consequently, any deterministic strategy needs an expected time of at least $s/2$ to determine a secret chosen uniformly at random. Since any randomized strategy is a convex combination of deterministic ones, this latter statement also holds for randomized strategies.
 
Let $S = (x^1, x^2, \dots)$ denote a deterministic strategy of Codebreaker. 
We first show a lower bound on the number of guesses that are needed to identify at least \emph{half} of all possible secret codes. For $j=1, \ldots, k^n$, let $A_j=A_j(S)\subseteq [k]^n$ denote the set of codes that can be uniquely determined from the answers to the queries $x^1, \ldots, x^j$. Let $s$ be the smallest index for which $|A_s|\geq k^n/2$.


Consider a code $Z \in [k]^n$ sampled uniformly at random, and set $Y_i := \eq(Z, x^i)$, $1 \le i \le s$.
Moreover, let
\[
\tilde Z =
\begin{cases}
Z & \text{ if $Z \in A_s$,} \\
\text{`fail'} & \text{ if $Z \notin A_s$.} \\
\end{cases}
\]
By our definitions, the sequence $Y:=(Y_1, Y_2, \dots, Y_s)$ determines $\tilde Z$, and hence by (E1) we have
\begin{equation} \label{eq:entropy-tilde-Z-Y}
H(\tilde Z) \le H(Y) .
\end{equation}

Moreover, we have
\begin{align}
H(\tilde Z)
& = -\sum_{z \in A_s} \Pr[ \tilde Z = z ] \log(\Pr[ \tilde Z = z]) - \Pr[ \tilde Z = \text{`fail'} ] \log(\Pr[ \tilde Z = \text{`fail'} ] \notag \\
& \ge - \sum_{z \in A_s} \Pr[ Z = z ] \log(\Pr[ Z = z]) \notag \\
& = \frac{|A_s|}{k^n} \log(k^n) \notag \\
& \ge \tfrac12 n \log k.
 \label{eq:lb-entropy-Z-tilde}
\end{align}

\ignore{
We derive a lower bound on $H(\tilde Z)$. Let $\mathbf{1}_{A_s}$ denote the indicator random variable for the event that $Z \in A_s$. Since $\mathbf{1}_{A_s}$ is determined completely by $\tilde Z$ we have by (E4) that
\begin{align} 
H(\tilde Z) & 
= H( \textbf{1}_{A_s}) + H( \tilde Z | \textbf{1}_{A_s}) \notag\\
& \ge H(\tilde Z | \textbf{1}_{A_s}) \notag \\
& \ge - \sum_{z \in [k]^n} \Pr[ \tilde Z = z \wedge Z \in A_s ] \log( \Pr[ \tilde Z = z | Z \in A_s ] ) \notag \\
& = - \sum_{z \in A_s} \frac1{k^n} \log \Big( \frac1{|A_s|} \Big) \notag \\
& = \frac{|A_s|}{k^n} \log (|A_s|) \notag\\
& \ge \frac{1}{2} \log \Big( \frac{1}{2} k^n \Big) = \Omega(n \log k). \label{eq:lb-entropy-Z-tilde}
\end{align}
}

We now derive an upper bound on $H(Y)$. For every $i$, $Y_i$ is binomially distributed with parameters $n$ and $1/k$. Therefore, its entropy is (see, e.g.,~\cite{jacquet1999entropy})
\begin{equation*}
H(Y_i) = \frac12 \log \left( 2\pi e \frac nk \Big(1 - \frac1k \Big) \right) + \frac12 + O \Big(\frac1n \Big) = O(\max\{\log(n/k), 1\}).
\end{equation*}
We thus obtain
\begin{equation} \label{eq:ub-entropy-Y}
H(Y) \stackrel{\textup{(E2)}}\le \sum_{i=1}^{s} H(Y_i) = s H(Y_1) = s O(\max\{\log(n/k),1\}).
\end{equation}
Combining~\eqref{eq:entropy-tilde-Z-Y}, \eqref{eq:lb-entropy-Z-tilde}, and~\eqref{eq:ub-entropy-Y}, we obtain
\[
s = \Omega \left( \frac{n \log k}{\max \{ \log(n/k), 1\} } \right).
\]
Since, by definition of $s$, at least half of all secret codes in $[k]^n$ can only be identified by the strategy $S$ after at least $s$ guesses, it follows that the expected number of queries needed to identify a uniformly chosen secret code is at least $s/2$.
\end{proof}

\subsection{Upper Bound for Non-Adaptive Strategies}

We first show that for $k = \Theta(n)$ a random guessing strategy asymptotically achieves the lower bound from Theorem~\ref{thm:nonadaptive}. Afterwards, we will show that one can also derandomize this.



\begin{lemma} 
\label{lem:random-guessing}
For black-peg Mastermind with $n$ positions and $k=\Theta(n)$ colors, the random guessing strategy needs an expected number of $O(n \log n)$ queries to determine an arbitrary fixed code $z\in[k]^n$. Furthermore, for a large enough constant $C$, $Cn\log n$ queries suffice with probability $1-o(1)$.

\end{lemma}

\begin{proof}
We can easily eliminate colors whenever we receive a $0$-answer. For every position $i \in [n]$ we need to eliminate $k-1$ potential colors. This can be seen as having $n$ parallel coupon collectors, each of which needs to collect $k-1$ coupons.

The probability that for a random guess we get an answer of $0$ is $(1-1/k)^n$, i.e., constant. Conditional on a $0$-answer, the color excluded at each position is sampled uniformly from all $k-1$ colors that are wrong at that particular position. Thus the probability that at least one of the $k-1$ wrong colors at one fixed position is \emph{not} eliminated by the first $t$ $0$-answers is bounded by $(k-1)(1-\frac{1}{k-1})^t\leq k e^{-t/k}$.

Let now $T$ denote the random variable that counts the number of $0$-answers needed to determine the secret code. By a union bound over all $n$ positions, we have $\Pr[T\geq t]\leq nk\e^{-t/k}=\Theta(n^2)\cdot e^{-\Theta(t/n)}$. It follows by routine calculations that $\mathbb{E}[T]=O(n\log n)$ and $\Pr[T\geq Cn\log n]=o(1)$ for $C$ large enough. As a random query returns a value of $0$ with constant probability, the same bounds also hold for the total number of queries needed.
\end{proof}


We now consider deterministic non-adaptive strategies to identify the secret code. \ifLongCites{\cite{Chvatal83}}{Chv\'atal~\cite{Chvatal83}} proved that the bound given in Theorem~\ref{thm:nonadaptive} is tight if $k \leq n^{1-\varepsilon}$, $\varepsilon>0$ a constant. 
Here we extend his argument to every $k \le n$. It essentially shows that a set of $O(\frac{n \log k}{\max \{ \log(n/k), 1\}})$ random guesses with high probability identifies every secret code. Our proof is based on the probabilistic method and is thus non-constructive. It remains an open question to find an explicit non-adaptive polynomial-time strategy that achieves this bound.

\begin{theorem} \label{thm:non-adaptive-lb}
There exists $n_0 \in \mathbb N$ and a constant $C > 0$ such that for every $n \ge n_0$ and $k \le n$ there exists a deterministic non-adaptive strategy for black-peg Mastermind with $n$ positions and $k$ colors that uses at most $C \frac{n \log k}{\max \{ \log(n/k), 1 \}}$ queries.
\end{theorem}


\begin{proof}
The idea is to use a probabilistic method type of argument, i.e., we show that, for an appropriately chosen constant $C > 0$ and $n$ large enough, a set of $N = C \frac{n \log k}{\max \{ \log(n/k), 1 \}}$ random guesses with positive probability identifies every possible secret code. (In fact, we will show that such a set of queries has this property with high probability.)

If a set $X=\{x^{(i)} \mid i \in N\}$ of queries distinguishes any two possible secret codes $z, z'$, then there must exist for each such pair $z\neq z'$ a query $x \in X$ with $\eq(z,x) \neq \eq(z',x)$. 
In particular we must have $| \{ i \in I(z,z') : x_i = z_i \} | \neq | \{ i \in I(z,z') : x_i = z'_i \} |$ for $I(z,z') := \{ i \in [n] : z_i \neq z'_i \}$.
Based on this observation we define (similar to~\cite{Chvatal83}) a \emph{difference pattern} to be a set of indices $I \subseteq [n]$ together with two lists of colors $(c_i)_{i \in I}, (c'_i)_{i \in I}$ such that $c_i \neq c'_i$ for every $i \in I$. 
For every two distinct secret codes $z, z' \in [k]^n$ we define the difference pattern corresponding to $z$ and $z'$ to be the set $I(z,z') := \{ i \in [n] : z_i \neq z'_i \}$ together with the lists $(z_i)_{i \in I}$ and $(z'_i)_{i \in I}$.
We say that a query $x \in [k]^n$ \emph{splits} a difference pattern given by $I$, $(c_i)_{i \in I}$, and $(c'_i)_{i \in I}$ if
\[
| \{ i \in I : x_i = c_i \} | \neq | \{ i \in I : x_i = c'_i \} |.
\]
It is now easy to see that if a set of $N$ queries has the property that every possible difference pattern is split by at least one query from that set, then these $N$ queries together with the answers deterministically identify Codebreaker's secret code.

In the following we show that a set of $N = C \frac{n \log k}{\max \{ \log(n/k), 1 \}}$ random queries with probability at least $1 - 1/n$ has the property that it splits every difference pattern.



The \emph{size} of a difference pattern $I$, $(c_i)_{i \in I}$, $(c'_i)_{i \in I}$ is the cardinality of $I$.
Note that for fixed $k$, the probability that a particular difference pattern is \emph{not} split by a randomly chosen query only depends on its size. Let $p(d,k)$ denote this probability for a difference pattern of size $d$. 
The probability that there exists a difference pattern that is not split by any of the $N$ random queries is at most
\[
\sum_{d=1}^{n} \binom nd (k(k-1))^d (p(d,k))^N .
\]
In order to show that this probability is at most $1/n$ it thus suffices to prove that for every $d \in [n]$ we have
\begin{equation} \label{eq:non-adaptive-ultimate-goal}
\binom nd (k(k-1))^d (p(d,k))^N < n^{-2} .
\end{equation}
We first take a closer look at $p(d,k)$. 
Observe that 
if a query $x$ does \emph{not} split a fixed difference pattern $I$, $(c_i)_{i \in I}$, $(c'_i)_{i \in I}$, then $x_i$ must agree with $c_i$ on exactly half of the positions in $I':=\{ i \in I \mid x_i \in \{c_i,c'_i\} \}$, and it must agree with $c'_i$ on the other positions in $I'$. In particular, the size of $I'$ must be even. More precisely, we have
\begin{align*}
p(d,k) & = \sum_{i=0}^{\lfloor d/2 \rfloor} \binom d{2i} \binom{2i}i \left( \frac 1k \right)^{2i} \left(1 - \frac2k \right)^{d-2i} \\
& = \sum_{i=0}^{\lfloor d/2 \rfloor} \binom d{2i} \left( \frac 2k \right)^{2i} \left(1 - \frac2k \right)^{d-2i} \binom{2i}{i} 2^{-2i} .
\end{align*}
Note that $\binom{2i}i 2^{-2i} \le 1/2$ for every $i \ge 1$, and $1-x \leq e^{-x}$ for all $x\in \R$. Hence,
\begin{align*}
p(d,k) & \le \left( 1 - \frac 2k \right)^d + \frac12 \sum_{j=1}^{d} \binom dj \left( \frac2k \right)^{j} \left( 1 - \frac2k \right)^{d-j} \\
& = 1 - \frac12 \left(1 - \left(1 - \frac2k \right)^d \right) \\
& \le \exp \left( - \frac12 \left(1 - e^{- \frac{2d}{k}} \right) \right).
\end{align*}
It follows that
\begin{equation} \label{eq:non-adaptive-pdk}
\ln \frac 1{p(d,k)} \ge \frac12 \left( 1 - e^{- \frac{2d}{k}} \right) .
\end{equation}
We now split the proof into two cases, $k \ge cn$ and $k < cn$ where $c$ is a sufficiently small constant. (We determine $c$ at the end of the proof.) 

\textbf{Case 1.} $k \ge c n$. Observe that in this case $\log(n/k) \le \log(1/c)$ and $\log k = \log n + \Theta(1)$. 
Hence, the bound claimed in Theorem~\ref{thm:non-adaptive-lb} evaluates to $O(n \log n)$ in this case. 
It thus suffices to show that there exists a constant $C > 0$ such that $N = C n \log n$ queries already identify every secret code with high probability.

We show $n^{5d}(p(d,k))^N < 1$ for every $d \in [n]$, which clearly implies~\eqref{eq:non-adaptive-ultimate-goal}. In fact, we show the equivalent inequality
\begin{equation} \label{eq:non-adaptive-case1}
\frac{N}{5d} \ln \frac{1}{p(d,k)} > \ln n.
\end{equation}
Using~\eqref{eq:non-adaptive-pdk} we obtain
\begin{equation} \label{eq:non-adaptive-case1-step1}
\frac{N}{5d} \ln \frac{1}{p(d,k)} \ge \frac{N}{10} \frac{1-e^{-\frac{2d}{k}}}{d}.
\end{equation}
Using that $d \mapsto (1-e^{-2d/k})/d$ is a decreasing function in $d$ we can continue with
\[
\frac{N}{5d} \ln \frac{1}{p(d,k)} \ge \frac{N}{10} \frac{(1-e^{-\frac{2n}{k}})}{n},
\]
which is clearly larger than $\ln n$ for any $N > \frac{10}{(1-e^{-2}) \log e} n \log n$. Hence for such $N$ we have~\eqref{eq:non-adaptive-case1} which settles this case.

\textbf{Case 2.} $k < c n$. In this case we need to be more careful in our analysis since in our claimed bound the factor $\log(n/k)$ might be large and the factor $\log k$ might be substantially smaller than $\log n$. 

In what follows, we regard only the case $k \ge 3$; the case $k=2$ has already been solved, cf.~\cite{Erd63}.

We first consider difference patterns of size $d \le \frac{n \log k}{\log(n/k) \log n}$. As in Case 1 we show that~\eqref{eq:non-adaptive-case1} holds for these patterns. Observe that~\eqref{eq:non-adaptive-case1-step1} holds again in this case. Since
the function $d \mapsto (1-e^{-2d/k})/d$ is 
decreasing in $d$ and since $d \le \frac{n \log k}{\log(n/k) \log n}$ we obtain
\begin{equation} \label{eq:non-adaptive-continue-bla}
\frac{N}{5d} \ln \frac{1}{p(d,k)} \ge \frac{N}{10} \frac{(1-e^{-\frac{2n \log k}{k \log(n/k) \log n}}) \log(n/k) \log n}{n \log k}.
\end{equation}
Next we bound the exponent $\frac{n \log k}{k \log(n/k) \log n}$ in the previous expression. Note that the derivative of $\frac{n \log k}{k \log(n/k) \log n}$ with respect to $k$ is
\begin{equation} \label{eq:non-adaptive-exponent}
\frac{ n ( \log n - \ln(2) \log k \log (n/k))}{\ln(2) k^2 \log n \log^2(n/k)}.
\end{equation}
We now show that this expression is 
less than 0 for 
$3 \le k \le n/4$. 
Indeed, observe that by setting $g(k) = \ln(2) \log k \log(n/k)$ we have for $n$ large enough that
$g(3) = \ln(2) \log(3) \log (n/3) > 1.09 \log(n) - 3.3 > \log n$
and 
$g(n/4) = 2\ln(2) \log(n/4)>\log n$. 
Moreover, observe that
\[
g'(k) = \frac{\log n - \log(k^2)}{k}.
\]
From this one easily sees that the function $g$ has a local maximum at $k=\sqrt{n}$ as its only extremal point in the interval in the interval $3 \le k \le n/4$.
Hence $g(k) > \log n$ for every $3 \le k \le n/4$ and thus \eqref{eq:non-adaptive-exponent} is negative.

Hence, $\frac{n \log k}{k \log(n/k) \log n}$ is a decreasing function in $k$ and we have $\frac{n \log k}{k \log(n/k) \log n} \ge \frac{1}{c \log(1/c)} \left(1 + \frac{\log c}{\log n} \right) \ge 1$ for $n$ large enough. With this we can continue~\eqref{eq:non-adaptive-continue-bla} with
\[
\frac{N}{5d} \ln \frac{1}{p(d,k)} \ge \frac{1-e^{-2}}{10} \frac{N \log(n/k) \log n}{n \log k},
\]
which is certainly larger than $\ln n$ for any $N \ge \frac{10}{(1-e^{-2}) \log e} \frac{n \log k}{\log(n/k)}$. This settles the case $k<cn$ for all $d \le \frac{n \log k}{\log(n/k) \log n}$.

In the remainder of this proof we consider the case $k<cn$ and $d \ge \frac{n \log k}{\log(n/k) \log n}$. For such $d$ we establish the inequality $2^n k^{2n} (p(d,k))^N < n^{-2}$ which clearly implies~\eqref{eq:non-adaptive-ultimate-goal}. As done previously, we actually show the equivalent inequality
\begin{equation} \label{eq:non-adaptive-case2b}
N \log \frac1{p(d,k)} > 2 n \log k + n + 2 \log n.
\end{equation}
First observe that $\binom{2i}{i} 2^{-2i} \le 1/\sqrt{i}$ for every $i \ge 1$. We denote by $\mathcal Bin(n,p)$ a binomially dsitributed random variable with parameters $n$ and $p$. With this, we obtain 
\begin{align*}
p(d,k) & \le \sum_{j = 0}^{\lfloor d/k \rfloor} \binom dj \left( \frac 2k \right)^j \left( 1 - \frac 2k \right)^{d-j} + \left( \frac dk \right)^{-1/2} \sum_{j = \lfloor d/k \rfloor + 1}^{n} \binom dj \left( \frac 2k \right)^j \left( 1 - \frac 2k \right)^{d-j} \\
& = \Pr \left[ \mathcal Bin \left( d,\frac 2k \right) \le \frac{d}{k} \right] + \left( \frac dk \right)^{-1/2} \sum_{j = \lfloor d/k \rfloor + 1}^{n} \binom dj \left( \frac 2k \right)^j \left( 1 - \frac 2k \right)^{d-j} .
\end{align*}
Using the Chernoff bound $\Pr[ \mathcal Bin(n,p) \le (1 - \delta) np] \le \left( \frac{e^{-\delta}}{(1-\delta)^{1-\delta}}\right)^{np}$ we obtain
\[
\Pr \left[ \mathcal Bin \left( d,\frac 2k \right) \le \frac{d}{k} \right] \le \left( \frac{e^{-1/2}}{\left( 1/2 \right)^{1/2}} \right)^{2d/k} = \left( \frac 2e \right)^{d/k}.
\]
Hence, we have
\[
p(d,k) \le \left( \frac 2e \right)^{d/k} + \left( \frac dk \right)^{-1/2} .
\]
It is not hard to see that the function
\[
f(k) = \frac{\left( \frac2e \right)^{d/k}}{\left( \frac dk \right)^{-1/2}}
\]
attains its maximum at $k = 2(1 - \ln 2)d$ and that $f(2(1-\ln 2)d) \le 1$. Hence we have
\[
p(d,k) \le 2 \left( \frac dk \right)^{-1/2} = \left( \frac{d}{4k} \right)^{-1/2}.
\]
With this we obtain
\begin{align*}
N \log \frac 1{p(d,k)} & \ge \frac N2 ( \log d - \log k - 2) \\
& \ge \frac N2 ( \log n + \log \log k - \log \log (n/k) - \log \log n - \log k - 2) \\
& = \frac N2 ( \log(n/k) - \log \log (n/k) - \log \left( \frac{\log n}{\log k} \right) - 2 ) \\
& \ge \frac N4 \log(n/k)
\end{align*}
where the last inequality follows from $\frac12 \log(n/k) - \log \log (n/k) - \log( \frac{\log n}{\log k} ) - 2 \ge 0$ for every $k \le cn$ for a sufficiently small constant $c > 0$ and $n$ large enough. (In fact, this step imposes the most restrictive bound on $c$, i.e., any $c > 0$ that, for $n$ large enough, satisfies $\frac12 \log(1/c) - \log \log (1/c) - \log( 1 - \frac{\log c}{\log cn}) - 2 \ge 0$ is appropriate for our proof.) Clearly $\frac N4 \log(n/k)$ is larger than $2n \log k + n + 2 \log n$ for any $N > 16 \frac{n \log k}{\log(n/k)}$ and $n$ large enough. This implies~\eqref{eq:non-adaptive-case2b} and thus settles this last case.
\end{proof}

\subsection*{Acknowledgments.} 
Carola Doerr is a recipient of the Google Europe Fellowship in Randomized Algorithms. This research is supported in part by this Google Fellowship. Her work also partially supported by a Feodor Lynen Research Fellowship for Postdoctoral Researchers of the Alexander von Humboldt Foundation and by the Agence Nationale de la Recherche under the project
ANR-09-JCJC-0067-01.

}


\end{document}